\documentclass{llncs}
\usepackage{latexsym}
\usepackage{amsmath}
\usepackage{amssymb}
\usepackage[table]{xcolor}
\usepackage{graphicx}
\usepackage{verbatim}
\usepackage{xspace}
\usepackage[tableposition=top]{caption}

\newcommand{\pspace}{\textsc{PSPACE}\xspace}
\newcommand{\np}{\textsc{NP}\xspace}
\newcommand{\conp}{\textsc{coNP}\xspace}
\newcommand{\tqbf}{\textsc{TQBF}\xspace}

\usepackage{array}
\begin{document}
\newtheorem{assumption}{Assumption}
\newtheorem{observation}{Observation}
\title{Deciding the On-line Chromatic Number of a Graph with Pre-Coloring is PSPACE-Complete}
\author{Christian Kudahl\thanks{Supported in part by the Villum Foundation and the Danish Council for Independent Research, Natural Sciences.}}
\institute{Department of Mathematics and Computer Science\\
University of Southern Denmark}
\maketitle
\begin{abstract}
The problem of determining if the on-line chromatic number of a graph is less than or equal to $k$, given a pre-coloring, is shown to be \pspace-complete.
\end{abstract}

\section{Introduction}
In the on-line graph coloring problem, the vertices of a graph are revealed one by one to an algorithm.
When a vertex is revealed the adversary reveals which other of the revealed vertices it is adjacent to.
The algorithm gives a color to the vertex. This color has to be different from all colors found on neighboring vertices.
The goal is to use as few colors as possible.

We let $\chi(G)$ denote the \emph{chromatic number} of $G$. This is the number of colors that an optimal off-line algorithm needs to color $G$.
Similarly, we let $\chi^O(G)$ denote the \emph{on-line chromatic number} of $G$. This is the smallest number of colors that the best on-line algorithm needs
to guarantee that for any ordering of the vertices, it will be able to color $G$ using at most $\chi^O(G)$ colors. It is useful to think of this
algorithm as knowing the graph in advance but not the vertex ordering. As an example, $\chi^O(P_4)=3$, since if two isolated vertices are presented
first, the algorithm will be unable to decide if it is optimal to give them the same or different colors. Clearly, $\chi(P_4)=2$.

The traditional measure of performance of an on-line algorithm is \emph{competitive analysis} \cite{Sleator}. Here, the performance of an algorithm is compared to the 
performance of an optimal off-line algorithm. In the on-line graph coloring problem, an algorithm $A$ is said to be \emph{$c$-competitive} if it holds,
that for any graph $G$, and for any ordering of the vertices in $G$, the number of colors used by $A$ is at most $c$ times the chromatic number of $G$.
For the on-line graph coloring problem, there does not exist $c$-competitive algorithms for any $c$ even if the class of graphs is restricted to trees \cite{gyarfas}.
This makes this measure less desirable to use in this context.

As an alternative \emph{on-line competitive} analysis was introduced for on-line graph coloring \cite{Gyarfas_on-linecompetitive}.
The definition is similar to competitive analysis, but instead of comparing
with the best off-line algorithm, one compares with the best on-line algorithm. In the case of on-line graph coloring, an algorithm is \emph{on-line $c$-competitive}
if for any graph $G$, and for any ordering of the vertices, the number of colors it uses is at most $c$ times the on-line chromatic number.

With the definition of on-line competitive analysis, a natural problem arose. How computationally hard is it given a graph $G$ and a $k \in \mathbb{N}$ to
decide if $\chi^O(G) \leq k$. In \cite{conjecture}, it was shown that it is possible in polynomial time to decide if
$\chi^O(G) \leq 3$ when $G$ is triangle free or connected. They conjectured it \np-complete to decide if $\chi^O(G) \leq 4$. In this paper, we consider the generalization of the problem where a part
of the graph has already been presented and colored (we refer to this as the pre-coloring). 
We show that it is \pspace-complete to decide if the rest of the graph, when presented in a on-line fashion,
can be colored such that a total of
at most $k$ colors is used for some given $k$. 

\section{Preliminaries}
On-line graph coloring can be seen as a game. The two players are known as the drawer and the painter. The two players agree on a graph $G=(V(G),E(G))$ and a $k \in \mathbb{N}$.
A move for the drawer is presenting a vertex (sometimes we say it request a vertex). 
It does not specify which vertex in $G$ the presented vertex corresponds to, but it specifies which of the already presented vertices that this new vertex is adjacent to.
The presented graph must always be an induced subgraph of $G$.

A move for the painter is assigning a color from $\{1,\ldots,k\}$ to the newly presented vertex. The color has to be
different from the colors that he previously assigned to its neighbors. If the painter manages to color the entire
graph, he wins. If he is ever unable to color a vertex (because all colors are already found on neighbors to this vertex) he loses.

When analyzing games, one is often interested in finding out which player has a winning strategy. A game is said to be \emph{weakly solved}
if it is known which player has a winning strategy from the initial position. It is said to be \emph{strongly solved} if it is known which player
has a winning strategy from any given position. This definition is the motivation behind the assumption to have a pre-coloring. We prove that to strongly
solve the game for a given graph, one must, in some cases, solve positions, where it is \pspace-hard to determine if the drawer or the painter has a win from that position.
Note that it may not be \pspace-hard to weakly solve the game - see closing remarks.

We consider the states in the game after an even number of moves. This means that the game has not started yet or the painter has just assigned a color to a vertex.
Such a \emph{state} can be denoted by $(G,k,G',f)$. Here, $G$ is the graph they are playing on and $k \in \mathbb{N}$ is the number of colors the painter is allowed to use.
Furthermore, $G'$ is the induced subgraph that has already been presented and colored and $f: \; V(G') \rightarrow \{1,\ldots,k\}$ is a function that describes what colors have been assigned to
the vertices of $G'$. 
Note that the painter does not get information on how to map the vertices of $G'$ into $G$ (in fact, the drawer does not have to decide this yet).

We treat the following problem: Let a game state $(G,k,G',f)$ be given. Does the painter have a winning strategy from this state?
We show that this problem is \pspace-complete. The problem is equivalent to deciding if the on-line chromatic number of $G$ is less 
than or equal to $k$ given that the vertices in an
induced subgraph isomorphic to $G'$
have already been given the colors dictated by $f$. This is also known as a pre-colored graph.

Note that the proof here also works in the model where the painter gets information on how the vertices in $G'$ are mapped to those in $G$. In fact, a slightly simpler construction
would suffice in that case. The model where this information is not available seems more reasonable though, since the pre-coloring is used to represent a state in the game
where this information is indeed not available.

We show a reduction from the totally quantified boolean formula (\tqbf) problem. In this problem, we are given a boolean formula:
\[
\phi= \forall x_1 \exists x_2 \; \ldots \; \exists x_n F(x_1,x_2,\ldots,x_n)
\]
We want to decide if $\phi$ is true or false. This problem is known to be \pspace-complete even if $F$ is assumed to be in conjunctive normal form with 3 literals in each clause (\cite{Stockmeyer19761}).
Since the complement to any language in \pspace is also in \pspace, this is also \pspace-complete if $F$ is in disjunctive normal form with 3 literals in each term (3DNF). This is the form we will
use here. We let $t_i$ denote the $i$'th term. For convenience, we will assume the number of variables to be even and that the first quantifier is $\forall$ followed by alternating quantifiers. This is possible since any 
\tqbf in 3DNF can be transformed to such a formula by adding new variables.

One such formula could for example be:
\[
	\forall x_1 \exists x_2 \forall x_3 \exists x_4 \; (x_1 \land x_2 \land \bar{x}_4) \lor (\bar{x}_1 \land x_2 \land x_3) \lor (\bar{x}_1 \land \bar{x_2} \land x_3)
\]
This formula has four variables, $x_1$, $x_2$, $x_3$, and $x_4$, and three terms, $t_1$, $t_2$, and $t_3$. The term $t_1$ contains $x_1$, $x_2$, and $\bar{x}_4$ (we also say that they \emph{are in} the first term).

\section{PSPACE Completeness}
In this section, we show that it is \pspace-complete to to decide if the painter has a winning strategy from a game state $(G,k,G',f)$.
First we note, that the problem is in \pspace.

\begin{observation} \label{obob}
The problem of deciding is the drawer has a winning strategy from state $(G,k,G',f)$ is in \pspace.
\end{observation}

To see this, we see that the game always ends within at most $2V(G)$ moves.
We need to show that from each state, the possible following moves can be enumerated in polynomial space.
If the painter is about the move, his possible moves are one of the colors $\{1,\ldots,k\}$.
This can be done is polynomial space, and can be enumerated based on the value of the color.
If the drawer is about to move, his move consists of presenting a vertex that is adjacent to some of the vertices that have already been presented.
If $v$ vertices have been presented already, this means that there are possibly $2^v$ different moves for him. He can enumerate these but only consider
those where the resulting graph is an induced subgraph of $G$. This problem is \np-complete, but it can be solved in polynomial space.

Using this, we do a post-order seearch in the game tree. In each vertex, we note who has a winning strategy from that given state.
For the leaves, we note who has won the game (done by checking if all vertices have been colored). After traversing the tree, we can read in the root
if the painter or the drawer has a winning strategy. This shows that the problem is in \pspace.

To prove that the problem is \pspace-hard,
we show how to transform a totally quantified boolean formula $\phi=\forall x_1 \exists x_2 \ldots \exists x_n \; F(x_1,\ldots,x_n)$ ($F$ is in 3DNF) into a game state $(G,k,G',f)$ such
that $\phi$ is true if and only if the painter has a winning strategy from $(G,k,G',f)$.

\begin{itemize}
\item The number of variables in $\phi$ is $n$.
\item The number of terms in $F$ is $t$.
\item We define $k=t+3n/2 + 2$ to be the number of colors that the painter is allowed to use. %jeg oegede k med 1 for at tage hoejde for farven som bruges i outside id gadget
\end{itemize}
We now describe $G=(V(G),E(G))$. We start by describing some graphs that are induced subgraphs of $G$. The graphs are $A=(V(A),E(A))$, $B=(V(B),E(B))$,
$H=((V(H),E(H))$, $X=(V(X),E(X))$, and $T=(V(T),E(T))$. For a $n \in \mathbb{N}$, we use $[n]$ to denote $\{1,\ldots,n\}$.

\begin{align*}
V(A) = & \{ a_i \; | \; i \in [k-3] \} \\
E(A) = & \{ (a_i,a_j) \; | \; i \neq j \} \\
V(B) = & \{ b_i \; | \; i \in [10n+3t] \} \\
E(B) = & \emptyset \\ \\
V(H_i) = & \{h_i^1, h_i^2, h_i^3, h_i^4 \} \\
E(H_i) = & \{(h_i^1,h_i^2), (h_i^2,h_i^3), (h_i^3,h_i^4) \} \\
V(H) = & \bigcup_{i=1}^{n/2} V(H_i) \\
E(H) = & \{ (v,w) \; | \; v \in V(H_i), w \in v(H_j), i \neq j \} \cup \bigcup_{i=1}^{n/2} E(H_i) \\
\end{align*}
\begin{align*}
V(T) = & \{ t_i \; | \; i \in [t] \} \\
E(T) = & \{ (t_i,t_j) \; | \; i \neq j \} \\
V(X) = & \{ x_i \;| \; i \in [n] \} \cup  \{ \bar{x}_i \; | \; i \in [n] \} \\
E(X) = & \{ (x_i,\bar{x}_i) \; | \; i \in [n] \}
\end{align*}

We now describe $G=(V(G),E(G))$. It contains $A$, $B$, $H$, $T$, and $X$ as well as some additional edges between them. It also contains two additional vertices.
After some lines, there are comments in parenthesis. They are not part of the construction, their purpose is to make to construction easier to read.

\begin{align*}
V(G)=& V(A) \cup V(B) \cup V(H) \cup V(T) \cup V(X) \cup \{m,c\} \\
E(G)=& E(A) \cup E(B) \cup E(H) \cup E(T) \cup E(X) \\
& \cup \{ (x_i,t_j) \; | \; \text{$x_i$ is in term $t_j$}, \; i \in [n], \; j \in [t] \} \\ 
& \cup \{ (\bar{x}_i,t_j) \; | \; \text{$\bar{x}_i$ is in term $t_j$}, \; i \in [n], \; j \in [t]  \} \\
&\mbox{ \text{(Connecting the variables to the terms)}} \\
& \cup \{ (h_i^j,x_{2i-1}) \; | \; i \in [n/2], j \in [4] \} \\
& \cup \{ (h_i^j,x_{2l}), (h_i^j,\bar{x}_{2l}) \; | \; j \in \{2,4\}, i \in [n/2], i \leq l \leq n/2 \} \\ 
& \cup \{ (v,w) \; | \; v \in V(H), w \in V(T) \} \\
& \mbox{ \text{(Connecting gadgets to variables and terms)}} \\
& \cup \{ (c,v) \; | \; v \in V(A) \cup V(B) \} \\
& \cup\{ (v,w) \; | \; v \in V(X), \; w \in V(A) \} \\
& \mbox{ \text{(Only two colors can be used on the variables)}} \\
& \cup \{ (m,v) \; | \; v \in v(T) \} \\
& \mbox{ \text{(Blocking color 1 from the terms)}} \\
& \cup\{ (b_j,x_{2i-1}),(b_j,\bar{x}_{2i-1}) \; | \; i \in [n/2], \; j \leq 3i-2 \} \\
& \mbox{ \text{(Identifies odd variables)}} \\
& \cup\{ (b_j,x_{2i}) \; | \; i \in [n/2], \; j \leq 3i-1 \} \\
& \mbox{ \text{(Identifies even variables)}} \\
& \cup\{ (b_j,\bar{x}_{2i}) \; | \; i \in [n/2], \; j \leq 3i \} \\
& \mbox{ \text{(Identifies even variables)}} \\
& \cup\{ (b_i,h_l^j) \; | \; l \in [n/2], \; i \leq l+3n/2, \; j \in [4] \} \\
& \mbox{ \text{(Identifies gadgets)}} \\
& \cup\{ (b_i,t_j) \; | \; j \in [t] \; i \leq j+2n \} \\
& \mbox{ \text{(Identifies gadgets)}} \\
\end{align*}

We now need to define $G'=(V(G'),E(G'))$ and $f: \; V(G') \rightarrow \{1,\ldots,k\}$ to complete the reduction.
\begin{align*}
V(G')=& \{ a'_i \; | \; i \in [k-3] \} \cup \\
& \{ b'_i \; | \; i \in [10n+3t] \} \cup \\
& \{ c', m' \} \\
E(G')=& \{ (c',a'_i) \; | \; i \in [k-3] \} \cup \\
& \{ (c',b'_i) \; | \; i \in [10n+3t] \} \cup \\
& \{ (a'_i,a'_j) \; | \; i \neq j, \; i,j \in [k-3] \}
\end{align*}
We define $f: \; V(G') \rightarrow \{1,\ldots,k\}$ in the following way.
\begin{align*}
f(a_i)&=i+3 \\
f(b_i)&=3 \\
f(c)&=1 \\
f(m)&=1
\end{align*}

Before continuing, we explain the main idea of the construction. 
The vertices in $X$ represent the variables and will only get color 1 or 2. This is ensured by connecting them to $A$ and $B$ where all but two colors are found.
The vertices in $T$ represent the terms. If no terms are satisfied (which has to be the case to make a formula in 3DNF false) there will
be at least one vertex in $T$ which cannot get a color.

The pre-coloring, $G'$, serves some purposes.
The purpose of $B$ is that counting the number of edges a vertex has to $B$ helps the painter identify which vertex he is painting.
With $A$, we are able to put restrictions on the colors the vertices in $X$ can get.
The only point of $c$ is to ensure that $G'$ can only be induced in $G$ by mapping $a'_i$ vertices to vertices in $A$ and $b'_i$ vertices to vertices in $B$.
Although $G'$ is isomorphic to the graph induced by $V(A) \cup V(B) \cup V(C)$, the painter gets no guarantee, that it is indeed these vertices. 
We start by showing, that there are no other
ways to induce $G'$ in $G$.
\begin{lemma} \label{id}
Let $g: V(G') \rightarrow V(G)$ be an injective function such that $(v,w) \in E(G')$ if and only if $(g(v),g(w)) \in E(G)$.
The following must hold:
\begin{itemize}
\item $g(c') = c$
\item $\forall i \in [k-3] \colon g(a'_i) \in V(A)$
\item $\forall i \in [2n+t] \colon g(b'_i) \in V(B)$
\item $g(m')=m$
\end{itemize}
\end{lemma}
\begin{proof}
No vertex in $G$ has a degree as high as that of $c$. This is true since no other vertex is adjacent to a vertex $b_i$ with $i>2n+t+1$
and there are only $2t+5.5n+1$ vertices outside $B$.
Because of this, it is only possible to map $c'$ (with just as high degree) to $c$.

The only neighbors of $c$ are those in $V(A) \cup V(B)$. Thus, the $a'_i$ and $b'_i$ vertices must be mapped to these.
It is easy to distinguish between vertices in $A$ and those in $B$, since $A$ is a complete graph and $B$ has no edges.
It now follows that $m'$ must be mapped to $m$, since this is the only vertex in $G$ that is not adjacent to any vertices in $V(B) \cup \{c\}$.
\qed
\end{proof}
We have now shown that the only possibility for inducing $G'$ in $G$ is mapping $c'$ to $c$, $m'$ to $m$, and the $a'_i$ and $b'_i$ vertices to
the $a_i$ and $b_i$ vertices respectively.

\begin{lemma} \label{id2}
When the drawer presents a vertex $v$, the painter is always able to identify an $i$ and which one of the following statements about $v$ holds
\begin{itemize}
\item $v$ is $\bar{x}_i$ and $i$ is even.
\item $v$ is either $x_i$ or $\bar{x}_i$ and $i$ is odd.
\item $v$ is $x_i$ and $i$ is even.
\item $v$ is $h_i^j$ for some $j \in [4]$.
\item $v$ is $t_i$.
\end{itemize}
\end{lemma}
\begin{proof}
Because of Lemma \ref{id}, the painter can count the edges between $v$ and vertices in $B$.
We call this number $j$. The following table shows how to compute $i$ from $j$ and how to decide
which one of the above statements holds.
\begin{center}
  \begin{tabular}{| >{\centering\arraybackslash}m{1in} | >{\centering\arraybackslash}m{1in} |>{\centering\arraybackslash}m{2in} |}%{ | p{4cm} | c | p{4cm} | }
    \hline
$j$ & $i$ & Statement about $v$ \\ \hline
\parbox[t]{6cm}{$j \leq 3n/2$, \\ $j \equiv 0 \pmod{3}$} & $i=\frac23 j$ & { \centering $v$ is $\bar{x}_i$ and $i$ is even } \\ \hline
\parbox[t]{6cm}{$j \leq 3n/2$, \\ $j \equiv 1 \pmod{3}$} & $i=\frac{2j+1}3$ & { \centering $v$ is either $x_i$ or $\bar{x}_i$ ($i$ is odd) } \\ \hline
\parbox[t]{6cm}{$j \leq 3n/2$, \\ $j \equiv 2 \pmod{3}$} & $i=\frac{2j+2}3$ & { \centering $v$ is $x_i$ and $i$ is even } \\ \hline
\parbox[t]{6cm}{$ 3n/2 < j \leq 2n$, } & $i=j-3n/2$ & { \centering $v$ is $h_i^j$ for some $j \in [4]$ } \\ \hline
\parbox[t]{6cm}{$ 2n < j $} & $i=j-2n$ & { \centering $v$ is $t_i$ } \\ \hline
  \end{tabular}
\end{center}
\qed
\end{proof}
We are now ready for the main proof. We begin with the easier implication.
\begin{lemma} \label{one}
If $\phi$ is false, then the drawer has a winning strategy from the state $(G,k,G',f)$.
\end{lemma}
\begin{proof}
We will call color 1 \emph{true} and the color 2 \emph{false}.
Since $\phi=\forall x_1 \exists x_2 \ldots \exists x_n \; F(x_1,\ldots,x_n)$ is false, it holds that
\[
	\exists x_1 \forall x_2 \ldots \forall x_n \; \lnot F(x_1,\ldots,x_n)
\]
This means that if two players alternately decide the truth values of $x_1,x_2,\ldots x_n$, there is a strategy $S$ for the player deciding the values of the
odd variables which makes $F$ false. The drawer is going to implement $S$.

The drawer will start by presenting the vertices in $X$ and $H$. It will do this in rounds.
In round $i$, $1 \leq i \leq n/2$,  it first presents $x_{2i-1}$ and $\bar{x}_{2i-1}$ in some order. It then presents the vertices of $H_i$ in some order.
Finally, it presents $x_{2i}$ and $\bar{x}_{2i}$ in some order. There are $n/2$ rounds. We want to show that the drawer can ensure that the following holds after round $i$:

\begin{itemize}
\item Each $H_j$ with $j \leq i$ has vertices with at least 3 different colors.
\item All $x_j$ and $\bar{x}_j$ with $j \leq i$ have been colored with colors true or false. 
\item Interpreting coloring as an assignment of truth values to variables $x_1,\ldots,x_i$, the drawer has a winning strategy in the game where the drawer and
the painter alternately decide the truth value for the remaining variables.
\item Either the colors true and false are not found in $H_i$ or the painter has lost the game.
\end{itemize}
For $i=0$, they all hold. Assume that they hold for some $i$.
We show that the drawer can present the vertices in round $i+1$ in an order that ensures that they hold after it.

The drawer starts by presenting $x_{2i-1}$ and $\bar{x}_{2i-1}$. Among the vertices that have already been presented (including those in the pre-coloring),
it holds that each vertex is either adjacent to both $x_{2i-1}$ and $\bar{x}_{2i-1}$ or none of them (note that $H_i$ has not been presented yet).
This means that the painter is unable to identify which is which.
Since they all both adjacent to all vertices in $A$ and some in $B$, the
only available colors for them are true and false. The painter has to assign true to one of them and false to the other.
The drawer now decides which one received the color true according to his winning strategy (we know that he has one from the induction hypothesis).
This ensures that he will have a winning strategy independent of whether variable $x_{2i}$ is set to true or false.

Now, the drawer presents two non-adjacent vertices from $H_i$.
The painter cannot identify which, since the vertices in $H_i$ are connected to the same vertices among those that have been presented.
If the painter gives these the same color, the drawer decides that they were $h_i^1$
and $h_i^4$. Otherwise, the drawer decides that they were $h_i^1$ and $h_i^3$. The drawer now presents the remaining two vertices of $H_i$ which results in
it containing at least three different colors. Note that the color 3 cannot be used in $H_i$, since all four vertices are adjacent to some vertices in $B$.

The drawer now presents $x_{2i}$ and $\bar{x}_{2i}$. According to Lemma \ref{id2}, the painter can identify which one is $x_{2i}$ and which one is $\bar{x}_{2i}$.
Again, the painter must color one true and the other false. This can be interpreted as the painter assigning a truth value to variable $x_{2i}$.
As we argued earlier, the drawer must still have a winning strategy if he decides the truth value of the remaining odd variables.

We need to argue that if a vertex in $H_i$ receives color true or false, the painter will immediately lose.
We first consider the case where color true is found on $h_i^2$ or $h_i^4$.
The painter will color $x_{2i}$ and $\bar{x}_{2i}$. These are adjacent to all vertices in $A$ (and some in $B$) meaning they can only get color true or false.
Since they are adjacent to $h_i^2$ and $h_i^4$, they cannot get the color true. Only color false is not available, and after one gets that,
the other cannot get any color. If the color true is found on $h_i^1$ or $h_i^3$ instead, the drawer changes the positions of $h_i^1$ and $h_i^4$
as well as those of $h_i^2$ and $h_i^3$. This is possible since it is not at this time possible for the painter to distinguish between $h_i^1$
and $h_i^4$ and between $h_i^2$ and $h_i^3$. This ensures that the color true does end up on $h_i^2$ or $h_i^4$ so we can argue in the same way.
The argument is similar if it is the color false is found in $H_i$.

This concludes the induction. We have now shown that after round $n/2$ all vertices in $X$ have been colored with colors true and false.
The truth assignment given to the variables in $X$ makes $F$ false. The drawer now presents all vertices in $T$ in any order. They cannot get the color true,
since they are adjacent to $m$ which has that color. Each vertex in $T$ is adjacent to a vertex in $X$ with color false since the truth assignment made $F$ false.
Furthermore, there are $3n/2$ colors that cannot be used on $T$ since they were used in $H$. Also, the color 3 cannot be used, since all vertices in $T$ are adjacent to
some in $B$. This leaves $k-2-3n/2-1=t-1$ colors. This is not enough to color the $t$ vertices, since they form a clique.
\qed
\end{proof}

We now show the other implication, which completes the proof.
\begin{lemma} \label{two}
If $\phi$ is true, then the painter has a winning strategy from the state $(G,k,G',f)$.
\end{lemma}
\begin{proof}
The painter has to color the part of $G$ that is not in $G'$ such that the resulting colored graph has at most $k$ different colors.
We notice that all remaining vertices have a least one neighbor in $B$ which means that the color $3$ is not available for any vertices.
This means that there are $k-1=3n/2+t+1$ colors left. Moreover, only the colors true and false are available for vertices in $X$.
We have already defined colors 1 and 2 to be called true and false. We call the next $3n/2$ colors the \emph{H-colors}. The $t-1$ remaining colors,
we call the \emph{T-colors}. The idea is that the $H$-colors will be used in $H$, true and false will be used in $X$ and the $T$ colors will be used in $T$.
Since there are only $t-1$ $T$-colors, the painter will needs to use true, false or an $H$-color on a vertex in $T$. This is only possible because $\phi$ is true. 

Before defining the painter strategy, we need a few preliminaries.
We start by defining \emph{normal play}.
In normal play, when a vertex $x_i$ or $\bar{x}_i$ with $i$ even is requested, the following must hold.
In each $H_j$ with $j \leq \frac{i}2$, both $h_j^1$ and $h_j^4$ have been requested.
We also define a \emph{good request}. A good request is a request to a vertex $t_i \in T$, where the following holds:
For each neighbor $v \in X$ of $t_i$, $v$ does not have color false and $v$'s neighbor in $X$ does not have color true.

For example, if $t_1$ was requested and $x_1$ was a neighbor, it would be a good request only if $x_1$ did not have the color false
(possibly because it had not been presented yet) and $\bar{x}_1$ did not have the color true (it might also not have been presented yet).
Note that when a vertex in $T$ is requested, the painter can identify if it is a good request using Lemma \ref{id2} and the fact that
it knows for each vertex in $T$ which neighbors in $X$ it has. 
We call it a good request because it results in the painter being able to use the color false on that vertex,
which means that he will have enough colors and win the game.

Since $\phi$ is true, there must exist a function $p$, which based on the truth assignment to $x_1,\ldots,x_{i-1}$
computes if variable $x_i$ ($i$ is even) should be true or false if the painter wants to make $F$ true. We define the function $p'$,
which computes if variable $x_i$ should be given color true or false if not all variables $x_1,\ldots,x_{i-1}$ have had their truth assignment
decided yet. For even $i$, we let $p'(x_i)=p(p'(x_1),\ldots,p'(x_{i-1}))$. For odd $i$, we define $p'(x_i)=\text{true}$ if $x_i$ has the color true,
if $\bar{x_i}$ has the color false or if none of them have been presented yet. We define $p'(x_i)=\text{false}$ otherwise. It useful to think of it
the following way: If $x_i$ is requested before $x_j$ and $\bar{x_j}$, $j < i$, the painter will be able to distinguish between $x_j$ and $\bar{x}_j$
when they get requested. Because of this, the painter just decides that $x_j$ is true and  it color $x_j$ and $\bar{x}_j$ accordingly when they get requested.

We now define a strategy for the painter. There are three phases. The painter starts in Phase 1. Certain events will cause the painter to enter Phase 2 or 3,
which in both cases means that the painter from that point can follow a simple strategy to win.
Table \ref{alg} defines how the painter handles a request to a vertex $v$ when in Phase 1. Phase 2 and 3 will be defined subsequently

\begin{center}
    \begin{table} \caption{Table defining a painter strategy} \label{alg}
    \begin{tabular}{ | c | p{2.5cm} | p{3cm} | p{6cm} |} 
    \hline
    Case & Subcase & Subsubcase & Color given \\ \hline
    $v \in V(H_i)$ & & & Color greedily with $H$-colors. \\ \hline
    $v \in V(X)$ & $i$ even & Normal play & Use color $p'(x_i)$ \\ \cline{3-4}
    & & Not normal play & Color greedily with \{true, false\} and go to phase 2. \\ \cline{2-4}
    & $i$ odd & No vertex in $H_{\frac{i+1}2}$ has been requested  & Color greedily with $\{\text{true, false}\}$. \\ \cline{3-4}
    & & At least one vertex in $H_{\frac{i+1}2}$ has been requested  & The painter can identify if the request is to $x_i$ or $\bar{x}_i$. 
	Use color true for $x_i$ and false for $\bar{x}_i$ \\ \hline

    $v \in V(T)$ & Good request & & Use color false and go to phase 3. \\ \cline{2-4}
    & Not good request & & Color greedily with $T$-colors \\ \hline
    \end{tabular} 
    \end{table} 
\end{center}
For now, we claim that if Phase 2 or Phase 3 are ever entered, the painter can win by following a simple strategy.

We show, that under normal play, the drawer will have to eventually make a good request, which makes the painter enter Phase 3.
First, we show that the truth assignment that $x_1,\ldots,x_n$ gets will make $F$ true. When an $x_i$ or $\bar{x}_i$ with even $i$ is requested, the painter will color it based
on the color of $x_1,\ldots,x_{i-1}$. However, since the drawer decides the order, it may happen that the truth values of these have not already been decided. For the variables with an even
index, this is not a problem for the painter, since it can just compute recursively, which color it will apply to it. For a variable with an odd index $x_j$, we defined that the painter
should consider it true (we set $p'(x_j)=\text{true}$ for odd $j$). This is possible since we are under normal play, which means that $h_\frac{j+1}2^1$ and $h_\frac{j+1}2^4$
have already been requested. When $x_j$ and $\bar{x_j}$ are requested, the painter is able to use this to see which one it is. According to Table \ref{alg}, the painter will give $x_j$ color
true and $\bar{x_j}$ color false which is exactly why it is possible for the painter to already consider $x_j$ as true before it has been requested, when $x_i$ is requested under normal play.
Note that $\phi$ is true. Since the painter colors according to $p$, the resulting truth assignment makes at least one term true.
This also gives, that at least one request to a vertex in $T$ will
be good (a request to $t_i \in T$ is not good if and only if term $t_i$ cannot be satisfied by the current truth assignment no matter what truth value the undecided variables are given).
We have now shown that the drawer must eventually make a good request under normal play. This shows that the game will either deviate from normal play at some point (making the
painter enter Phase 2 or make a good request making the painter enter Phase 3. We now define how the painter behaves in Phase 2 and Phase 3 and show why he will win in both cases.

At the beginning of Phase 2, the drawer has just deviated from normal play. He has presented $x_i$ (or $\bar{x}_i$) with $i$ even, even though there exists a $H_j$ with $j \leq \frac{i}2$ where $h_j^1$ and $h_j^4$ have
not both been requested. Note that $H_j$ is bipartite (it is a $P_4$). Since $H$ was colored greedily, and $h_j^1$ and $h_j^4$ have not both been presented,
we know that at most one color is already used in each partition and no color is already used in both partitions.
For future requests in $H_j$, the painter will know which partition the requested vertex is in, since $h_j^2$
and $h_j^4$ are connected to the vertex in $X$ that was just requested. Thus, the painter will only have to use 2 colors for $H_j$. For the remaining requests, the painter colors greedily with $H$-colors
in $H$. He colors greedily with $\{\text{true,false}\}$ in $X$ and he colors greedily with $T$-colors in $T$. When the final vertex in $T$ is requested, there will not be a $T$-color available
(since there are only $t-1$). However, the painter will have one $H$-color that is not needed (the color saved in $H_j$). He uses that as a $T$-color, which ensures that he wins.

At the beginning of phase 3, the painter has just assigned the color false to a vertex in $T$ after a good request. 
Since the request was good, we know that all adjacent vertices in $X$ have been or can be colored true. Their neighbours in $X$ have been or can be colored false.
The remaining vertices in $X$ get colored greedily with $\{\text{true,false}\}$.
The vertices in $H$ will be colored greedily using $H$-colors. The remaining vertex in $T$ will be colored greedily using $T$-colors
which suffices. This ensures, that the painter wins.

We have now presented a strategy for the painter. We have shown that either Phase 2 or Phase 3 will always be entered and we have shown
how the painter wins once such a Phase has been entered.
\qed
\end{proof}

We can now combine Lemmas \ref{one} and \ref{two} and Observation \ref{obob} to get the desired theorem.
\begin{theorem}
Given a state $(G,k,G',f)$ in the on-line graph coloring game, it is \pspace-complete to decide if the painter has a winning strategy.
\end{theorem}

\section{Closing remarks}
The complexity of the problem of deciding if $\chi^O(G) \leq k$ is still open. It was shown to be \conp-hard in \cite{kudahl},
and it is certainly in \pspace using the argument presented here.
Adding a pre-coloring ensures that the problem is \pspace-complete.
Our work with the problem has led to the following conecture:
\begin{conjecture}
Let a graph $G$ and a $k \in \mathbb{N}$ be given. The problem of deciding if $\chi^O(G) \leq k$ is \pspace-complete.
\end{conjecture}
It seems likely to us, that a reduction from totally quantified boolean formula in 3DNF is possible.
It may be possible to use a similar construction to the one used here, but special attention has to be given to the case where $\phi$ is true.
It is challenging to implement the winning strategy from the satisfiability game when the drawer is able to request any vertex in the graph
without the painter knowing which vertex is being requested.

\bibliographystyle{splncs03}
\bibliography{refs}

 \end{document}